\documentclass[preprint]{iucr}              

     \journalcode{J}              

\usepackage{verbatim}
\usepackage{color,soul}
\usepackage{amsmath}

\usepackage[utf8]{inputenc}
\usepackage[english]{babel}
\usepackage{amsthm}
\newtheorem{theorem}{Theorem}[section]

\newtheorem{lemma}[theorem]{Lemma}
\theoremstyle{definition}
\newtheorem{definition}{Definition}[section]

\usepackage{amssymb}

\begin{document}                  



\title{Super-resolution SAXS based on PSF engineering and sub-pixel detector translations }


\author[a]{Benjamin}{Gutman}
\author[b]{Michael}{Mrejen}
\author[a]{Gil}{Shabat}
\author[b]{Ram}{Avinery}
\author[a]{Yoel}{Shkolnisky}
\cauthor[b]{Roy}{Beck}{roy@tauex.tau.ac.il}

\aff[a]{School of Mathematics Sciences}
\aff[b]{the Raymond and Beverly Sackler School of Physics and Astronomy, Tel Aviv University, Tel Aviv 69978, Israel}









\maketitle                        

\begin{synopsis}
We develop a two-step reconstruction methodology to enhance the angular resolution for  given  experimental  conditions.  Using  minute  hardware  additions,  we  show  that translating  the  x-ray  detector  in  subpixel  steps  and  modifying  the  incoming  beam shape results in a set of 2D scattering images which is sufficient for super-resolution SAXS  reconstruction. 
\end{synopsis}

\begin{abstract}

Small-angle X-ray scattering (SAXS) technique enables convenient nanoscopic characterization for various systems and conditions. Nonetheless, lab-based SAXS systems intrinsically suffer from insufficient x-ray flux and limited angular resolution. Here, we develop a two-step reconstruction methodology to enhance the angular resolution for given experimental conditions. Using minute hardware additions, we show that translating the x-ray detector in subpixel steps and modifying the incoming beam shape results in a set of 2D scattering images which is sufficient for super-resolution SAXS reconstruction. The technique is verified experimentally to show above 25\% increase in resolution. Such advantages have a direct impact on the ability to resolve faster and finer nanoscopic structures and can be implemented in most existing SAXS apparatuses.
\end{abstract}

\section{Introduction}

The scientific and applicative revolution in nanotechnology over the past decade requires nanoscale characterization techniques, which are versatile and technologically challenging. Small-angle X-ray scattering (SAXS) is an established technique that is reviving and flourishing in the last decade. Traditionally, SAXS was used to characterize nanoscale objects at low-resolution in various media in a non-destructive way \cite{KORNREICH2013716,blanchet2013small,Jacoby2015,Safinya2013,Ben-Nun2010,Lipfert}. Most SAXS experiments are currently conducted in specialized synchrotron facilities, that harness the high-flux of x-ray photons to provide adequate scattering signal. Such experiments are not accessible to most industries and scientists, as synchrotron-based experiments are commonly over-booked world-wide. 

However, recent advances in small scale lab-based x-ray sources and, in particular, in high-efficiency solid-state 2D, x-ray detectors enable conducting many of the experiments using lab-based SAXS systems, resorting to synchrotron experiments in only specific cases. Currently, there are a handful of well-established SAXS suppliers that provide systems that can measure structural information in a variety of disciplines. Nevertheless, since the pixel size of most x-ray detectors is relatively large, current SAXS systems require significant floor space. A large sample to detector distance (typically held in a vacuum) is required to achieve high angular resolution in the small-angle region.

For a point source, the SAXS signal, $I(q)$, measured at a sample to detector distance $d_s$, and at position $\Delta r$ from the unscattered direct beam position, is a Fourier transform squared of the sample's electron density. Here, $q = 4\pi \sin{\theta}/\lambda$ is the scattering wave vector, $\lambda$ is the X-ray wavelength, and $\theta$ is the scattering angle. From geometrical consideration,  $\tan 2\theta=\Delta r / d_s$. Therefore, enhanced angular resolution requires either large $d_s$ or excellent sampling of $\Delta r$ facilitated by a small pixel size of the X-ray camera.        

Moreover, similar to conventional imaging techniques, the measured SAXS intensity, $I_m$, is a convolution of the sample scattering intensity with the point-spread function (PSF) of the system, namely:
\begin{equation}\label{eq:PSF_conv}
    I_m(\langle q \rangle)=\int P(q,\langle q \rangle)*I(q)dq.
\end{equation}

Here, $\langle q \rangle$ is the average scattering vector corresponding to the setting of the instrument, and $P(q,\langle q \rangle$) is the PSF determined by the wavelength spread, finite direct beam collimation, and the detector's resolution. It is worth noting that for small angles ($q \rightarrow 0$), the collimation-related PSF, which depends on the actual system configuration, can be measured directly from the shape of the direct beam \cite{Pedersen:1990el}. 

From Eq. \ref{eq:PSF_conv}, it is evident that smaller PSF will yield a better resolution. However, due to limited flux, this entails a significant increase of the time needed for the experiment in order to achieve sufficient signal to noise ratio (SNR).  

Nowadays, PSF deconvolution in SAXS data is performing poorly  \cite{Vad:ks5259}. In most cases, calculated SAXS models are convoluted with measured or guessed PSF to fit the experimental data \cite{Hua_2017}. This procedure is not ideal as it introduces bias into the reconstructed model. Recently, interactive procedures have been proposed to de-smear $I_m(q)$ and obtain the pure scattered intensity data $I(q)$. However, these methods had limited success in noisy conditions \cite{Vad:2010hu}.

Here, we propose a two-step technique that improves the angular resolution by more than 25\%. Data recording and analysis schemes are tested on existing SAXS apparatuses. We show that minor hardware modifications are allowing to achieve super-resolution SAXS (srSAXS) and to probe nanoscale materials in lab-based systems. Our approach is based on the implementation and integration of translation based subpixel reconstruction and multi-PSF engineering and reconstruction algorithms. 

For the first phase, inspired by super-resolution microscopy \cite{UR1992181}, we control and monitor the translations of the X-ray detector. We use a sub-pixel reconstruction method based on \cite{Farsiu:2004}, and achieve smaller effective pixel size on a commercial 2D X-ray detector. 

In the second phase, motivated by the new solid-state x-ray detectors and recent development of scatterless slits design \cite{Li:2008cx}, we control measure and analyze the direct beam profile to estimate the blurring PSF. Following, we measure the scattering profiles with different PSFs and reconstruct an optimized X-ray diffraction pattern that resolves SAXS data with superior resolution.

Future applications of the suggested advances can have a tremendous impact on various aspects and industries as it will provide top-of-the-line nanoscale characterization techniques of quality comparable with that of synchrotron-based experiments.
The proposed lab-based srSAXS system allows us to measure smaller angles for a given sample-to-detector distance or to dramatically reduce the sample-to-detector distance for similar small-angle separation resolution and, therefore, extensively improve existing lab-based SAXS resolution and dynamic-range capabilities. 

\section{Methods}

\subsection{Synthetic data generation} \label{Synthetic_data_generation}
Three premises are taken in order to generate synthetic scattering images: (a) the ground truth (GT) scattering pattern has a circular symmetry with a Lorentzian profile. (b) Photon measured by the detector has Poisson statistics. Furthermore, (c) some spurious background scattering pattern exists, arising regardless of the sample measured, such as scattering of the beam stop, of the sample holder, the direct (unscattered) beam, cosmic radiation, etc. Summing the above, we place two close Lorenzian circles, add background signal that is reciprocal to the distance from the center (i.e., $q$), place beam-stop for the singularity in the center and run it through Poisson process:

\begin{equation}
\label{eq:data_simulation}
    I[x, y] = poissrnd(t \cdot I_{GT}[x, y]).
\end{equation} 
Here, $poissrnd$ is implemented using Matlab for random Poisson generator, $t$ is the simulated exposure time, and $I_{GT}$ is the ground-truth scattering profile. 

For simplicity, we model the PSF as a convolution of point Gaussian source and the rectangular shape of scatterless slits \cite{Li:2008cx}. Since the convolution of Gaussian and rectangle is a sum of two error-functions, the PSF is modeled as such. Finally, the different PSFs were convoluted with the GT image in order to create synthetic images. Additional details are given in the appendix \ref{app:syn_PSF}.

\subsection{Samples preparation}
Commercial AgBh powder (Thermo Fisher Scientific) was used without any further purification. 1,2-dioleoyl-sn-glycero-3-phosphoethanolamine (DOPE) was purchased from Avanti Polar Lipids Inc. The lipids were dissolved in water (DDW), total lipid concentration was 30 mg/ml per sample. Samples were homogenized using a vortexer for 5 minutes at 3000 RPM. Samples were then placed in quartz capillaries (1.5 mm in diameter), containing about $40 \mu l$.

\subsection{SAXS measurement setup}

Measurements were performed using a lab-based X-ray scattering system, with a GeniX (Xenocs) low divergence Cu $K_{\alpha}$ radiation source (wavelength of $\lambda=1.54$\AA) and a scatterless slits setup \cite{Li:2008cx}. Samples were measured at distance of $d_s = 117 mm$ using Pilatus 300K detector (Dectris). The detector, sample stage, and slits were motorized using stepper motors with a positioning accuracy of 1 $\mu m$ and controlled by SPEC software.

\section{Results and discussion}

In SAXS instruments, a monochromatic beam of X-rays illuminates a sample from which some of the X-rays scatter, while most go through the sample without interacting with it. The scattered X-rays are collected by a camera (typically a 2D flat X-ray detector) situated behind the sample. However, most available lab-based X-ray sources produce divergent beams and thus rely on collimating the direct beam. In our SAXS system, we control the beam shape and flux using scatterless slits design\cite{Li:2008cx}, allowing the engineering of the response function of the system.

Our super-resolution method is composed of two consecutive phases:(a) subpixel translation and (b) multi-PSF engineering. Below, we will show that for a budgeted experimental time, a specific experimental protocol results in superior angular resolution. We compare our results to alternative measurement protocols.    

\subsection{Subpixel sampling}
X-ray detectors suffer from relatively large pixel dimensions. For a given experimental setup, this significantly limits the achievable angular resolution, in particular at the smallest angles. 
To reconstruct the scattering with subpixel resolution, we conduct several measurements of a given sample, each time with a different position of the camera. Between consecutive recordings, we translated the detector position with subpixel intervals. For a given pixel size ($l$), we define the resolution enhancement factor ($f$), for which the sub-pixel translation is of $l/f$. This results in overlapped pixels in the registered images, where each pixel contains partial information about the higher resolution data. 

To extract the high-resolution image, we first find the actual relative displacement between the different recordings using the Lucas-Kanade algorithm \cite{Lucas:1981}, and then apply a super-resolution (SR) algorithm for proper image fusion \cite{Farsiu:2004}.

We applied this approach to synthetic data modeling realistic noise and X-ray diffraction pattern (Figure \ref{fig:SR_synthetic}) with $f=3$ (9 translations on 3x3 grid). The results demonstrate an improved resolution and our ability to decipher two nearby peaks otherwise undetectable given the original pixel dimension. 

Intuitively, increasing the number of recordings with smaller subpixel translation will yield a higher resolution. However, provided a budgeted time for a given experiment, there is a trade-off between the possible improved resolution and the SNR for each recording. Moreover, the application of the super-resolution method for a larger translated grid is computationally expensive, and with limited benefits in the reconstructed 2D image. 

\begin{figure}
\includegraphics[width=0.9\textwidth]{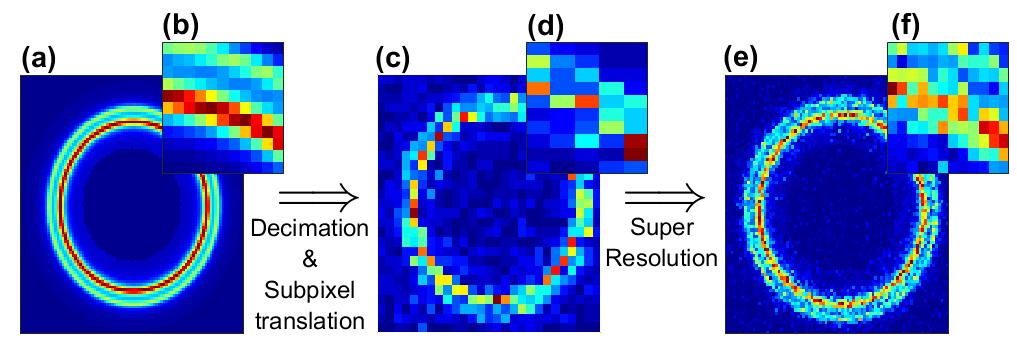}
\caption{Super-resolution process for synthetic image: (a) A synthetic high resolution ground truth (GT) of two nearby circular peaks, (b) down-sampled translated image of the GT for realistic measurement recording, and (c) super-resolution reconstruction. Insets show zoomed-in images.}
\label{fig:SR_synthetic}
\end{figure}

Motivated by the synthetic results, we gathered experimental data of X-Ray scattering patterns in a short sample to detector distance ($d_s=117 mm$). Typically, SAXS is measured with an order of magnitude larger $d_s$. Our test samples are a Silver behenate (AgBh) powder, showing lamellar scattering signals, and DOPE phospholipids in solution, showing self-assembled inverted hexagonal phase. Both samples display isotropic scattering rings conveniently analyzed in the azimuthally integrated signal with respect to the direct beam center. 

Each of the measured samples provides a different perspective of the super-resolution reconstruction. AgBh allows investigating the sharpening of scattering features at small angles. On the other hand, the DOPE sample has two close scattering peaks giving a more tangible appreciation of the resolving power improvement of the reconstruction algorithm (Figure \ref{fig:true_SR}). 

We notice that the direct beam profile, and therefore the system's PSF have a significant effect on the ability to achieve higher resolution reconstructions. We define separation criteria of two nearby peaks, such as $q_{H(1,1)}$ and $q_{H(2,0)}$ in the DOPE scattering, by:
\begin{equation}
\label{eq:resolution_crit}
    \delta=\frac{I_p-I_v}{\Delta q}.
\end{equation} 

Here, $I_p$ and  $I_v$ are the intensities of lower peak and valley between the peaks, respectively, and $\Delta q$ is the distance between them.   
For example, for beam profile of $0.6\times0.6$ $mm^2$ results in $\delta_{0.6\times0.6} = 0.60$, while for larger beam profile ($1\times1 mm^2$) $\delta_{1\times1} = 0$ since no distinction of close peaks is observed.

In Figure \ref{fig:true_SR}, we present the results of implementing the SR algorithm on AgBh and DOPE measurements. The samples were taken with f=3 and with various PSFs. We find that the main resolution enhancement contribution is in denoising the image and reducing the width of the circle in the AgBh measurement: The width of the central circle is reduced by $\sim 10\% $ in all examined measurements while the $\delta$ criteria (Eq. \ref{eq:resolution_crit}) showed no significant improvement.

\begin{figure}
\includegraphics[width=0.99\textwidth]{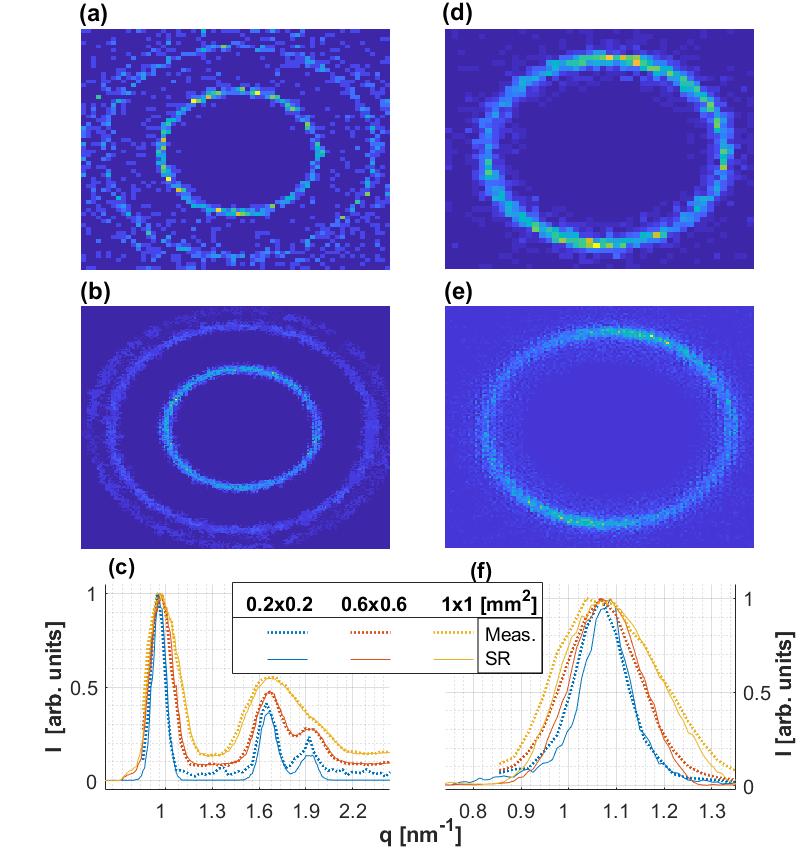}
\caption{Resolution performance of the SR algorithm for (a-c) DOPE and (d-f) AgBh samples. 
(a, d) Measured  with $0.2 \times 0.2 mm^2$ beam profile.
(b, e) Result of the SR algorithm with $0.2 \times 0.2 mm^2$ beam profile. 
(c, f) 1D azimuthal integration SAXS profiles and reconstructions with 3 different beam sizes. }
\label{fig:true_SR}
\end{figure}

\subsection{Constrained Multi-Deconvolution (CMD)}
\label{sec:CMD}

\begin{figure}
\includegraphics[width=0.99\textwidth]{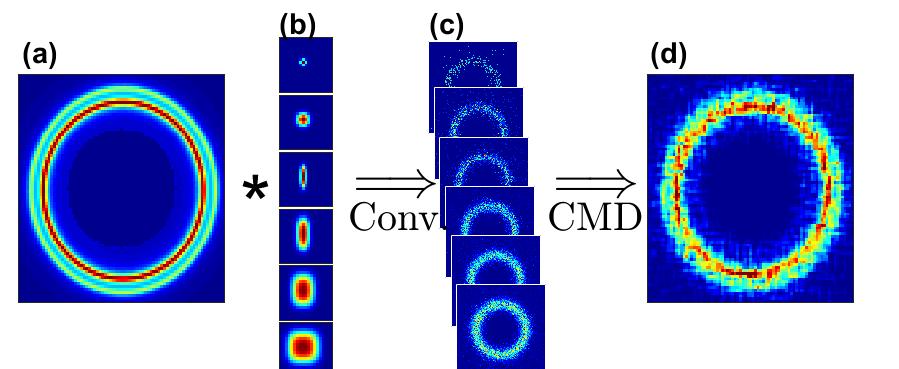}
\caption{CMD process for synthetic data. (a) Synthetic data used to simulate ground truth (GT),
(b) Synthetic PSF used to blur the GT. (c) Convoluted images representing blurred measurements. (d) The resulting reconstructed image after the CMD algorithm.}
\label{fig:CMD_synthetic}
\end{figure}

The second most dominant resolution limiting factor is the finite size of the direct X-ray beam. While smaller beams will result in a higher angular resolution, it will impose much longer exposure times in order to achieve comparable SNR. Constrained by insufficient flux at smaller beam profiles and resolution at larger ones, we develop an algorithm that takes advantage of both. 
Our new reconstruction algorithm, coined as Constrained Multi-Convolution (CMD), is based on measurements with different PSFs on the same sample. For a set of $m$ measured images, $Y_i$, and their corresponding measured PSF - $P_i$, where $i=1,...,m$, our goal is to restore the image as if the PSF was a point source (i.e., delta function). The outcome of such a process will be a de-blurred and de-noised image (Figure \ref{fig:CMD_synthetic}). 

Formally, we model each measured image as a convolution of the "real" image, X, we wish to restore and a PSF: 
\begin{equation}
    Y_{i}=X*P_i.
\end{equation}

A direct least squares solution can be formulated as following:
\begin{equation}\label{eq:cmd_first}
    \hat{X}=\arg\min_{X}\sum_i||Y_i-P_i*X||_F^2,
\end{equation}

where $||.||_F$ refers to the Frobenius norm: $||A||_F= \sqrt{\sum_{i, j} A_{i, j}^2}$. 

One way of dealing with the resolution to SNR trade-off is phrasing a weight function that controls the significance of each image $Y_i$ in Eq. \ref{eq:cmd_first}. The weight function should depend on the underlying scattering image, as the following thought experiment demonstrates. For a scattering pattern requiring high resolution (e.g., delta function in q-space), large PSF will blur the image and, therefore, must not be weighted highly into the algorithm. On the other hand, if the resolution of the available PSF is of the underlying scattering signal (or finer) than the small PSF images will only suffer from low SNR and thus will corrupt the reconstruction process. 

We therefore denote by $\sigma$, a weight function, and reformulate the CMD reconstruction algorithm to:

\begin{equation}\label{eq:CMD_second}
    \hat{X} = \arg\min_{X}{\sum_i\frac{||Y_i-P_i*X||_F^2}{2\sigma_i^2}+\lambda||X||_F^2}.
\end{equation}

Here, in addition to $\sigma$, we found that a ridge regularization is required in order to ensure a numerically stable solution. This regularizing term takes the form of balancing the cost function with the target matrix norm \cite{Golub:2013}.

Representing the convolution by the corresponding matrix multiplication, Eq. \ref{eq:CMD_second} can be solved as a linear system. However, the corresponding matrices are much larger than the actual image, which results in a heavy processing and computational load. For example, convolution of $n\times n$ sized image with a $m\times m$ sized PSF transforms to an $(m+n-1)\times(n^2)$ matrix. 
In order to alleviate the computational load, let us rephrase Eq. \ref{eq:CMD_second} in operator notation. 

For a set of measured PSFs, we define $\mathcal{P}$ to be an operator acting on images by $\mathcal{P} = [P_1* \quad P_2* \quad...\quad P_m* \quad \sqrt{\lambda}I]^T$.
In a similar way, we stack the measured images as $\quad \mathcal{Y} = [Y_1 \quad Y_2 \quad ... \quad Y_m \quad 0]^T$.

In this notation, Eq. \ref{eq:CMD_second} becomes

\begin{equation} \label{eq:CMD_operator}
    f(X) = ||D(\mathcal{P}X - \mathcal{Y}||_F^2 = \left\Vert D\left( \begin{array}{cc} P_1*X & - Y_1 \\ P_2*X & - Y_2 \\ \vdots & \vdots \\ P_m*X & -Y_m \\ \sqrt\lambda I & 0 \end{array} \right) \right\Vert_F^2 .
\end{equation}

Here, $D$ is a diagonal matrix representing $\sigma$ from Eq. \ref{eq:CMD_second}: 
\begin{equation} \label{eq:D_operator_1}
D = \text{diag}(\sigma_1 I \quad \sigma_2 I \quad ... \quad \sigma_m I \quad I).
\end{equation}

Differentiating Eq. \ref{eq:CMD_operator} with respect to $X$ we get the normal equations 
\begin{equation} \label{eq:f_grad}
\bigtriangledown{f} = (D\mathcal{P})^T D(\mathcal{P}X - \mathcal{Y}),
\end{equation}
where $ \mathcal{P}^T$ is the conjugate operator in the Frobenius inner product notation. Requiring $ \bigtriangledown{f} = 0$ and using the diagonal form of D, we find our optimum in the form of:

\begin{equation} \label{eq:normal_equation}
   \mathcal{P}^T D^2 \mathcal{P}X = \mathcal{P}^T D^2 \mathcal{Y}
\end{equation}

This equation is solved using the conjugate gradients method, detailed in Appendix \ref{apn:conjugate_proff}. We applied the CMD approach to synthetic data, modeling realistic noise, and X-ray diffraction pattern (Fig. \ref{fig:CMD_synthetic}c).

\begin{figure}
\includegraphics[width=0.95\textwidth]{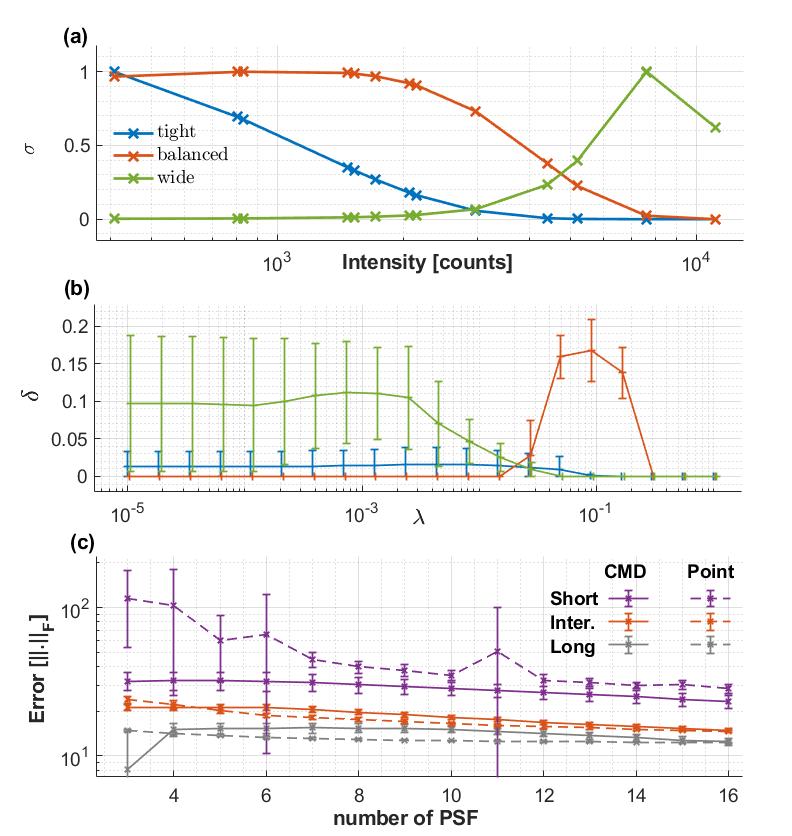}
\caption{
Reconstruction parameters optimization. 
(a) Three different $\sigma$ weight functions examined for various total intensity in the simulated data: emphasizing tight (blue),  balanced (red), and wide (green) PSFs. Total intensity is correlated with specific PSF area.  
(b) Performance of the CMD algorithm ($\delta$, defined in Eq. \ref{eq:resolution_crit}) for various regularization factors ($\lambda$) and $\sigma$ weight functions. 
(c) Reconstruction's error for CMD with increasing number of PSFs taken in account (solid lines) using the optimal $\lambda = 0.1$ and balanced $\sigma$ function (red curve in a). Reconstructed scoring (y-axis), is calculated using the difference of the reconstruction from the ground-truth image (Fig. \ref{fig:CMD_synthetic}a) in Frobenius norm. For comparison, a single point-like PSF image (labeled as 'point') with equivalent exposure is presented in dashed lines. The different lines represent short, intermediate and long exposure times simulation with t = 0.1, 0.8 and 4 respectively, as defined in Eq. \ref{eq:data_simulation}. Error bars are standard deviation of 10 independent simulations and reconstructions.}
\label{fig:CMD_params}
\end{figure}

An optimal CMD performance depends on several factors: $\lambda$, the duration of the experiment, the number and geometry of PSFs used by the CMD, and  $\sigma$ for each PSF. From our experience, while bounds to these parameters can be obtained in general, optimal reconstruction parameters depend on the sample's scattering under investigation. In the following, we will demonstrate that with optimized parameters, CMD outperforms measurements with the smallest PSF at equivalent an total measurement time.

We generated synthetic X-ray scattering data with two nearby scattering rings (Fig. \ref{fig:CMD_synthetic}a) and demonstrate the relation between the reconstruction parameters and the resulting resolution (Fig. \ref{fig:CMD_params}). For three different $\sigma$ functions (Fig. \ref{fig:CMD_params}a), reconstruction's resolution has significantly changed with optimal $\lambda$ (Fig. \ref{fig:CMD_params}b). In this example, for optimal reconstruction the weights of smaller PSF's are elevated in comparison to larger ones. Moreover, for a fixed set of PSFs, we find mild dependence between the exposure time (or SNR) and $\lambda$ resulting with the best reconstruction parameters (i.e. larger $\delta$  in Eq. \ref{eq:resolution_crit}). 

\begin{figure}
\includegraphics[width=1\textwidth]{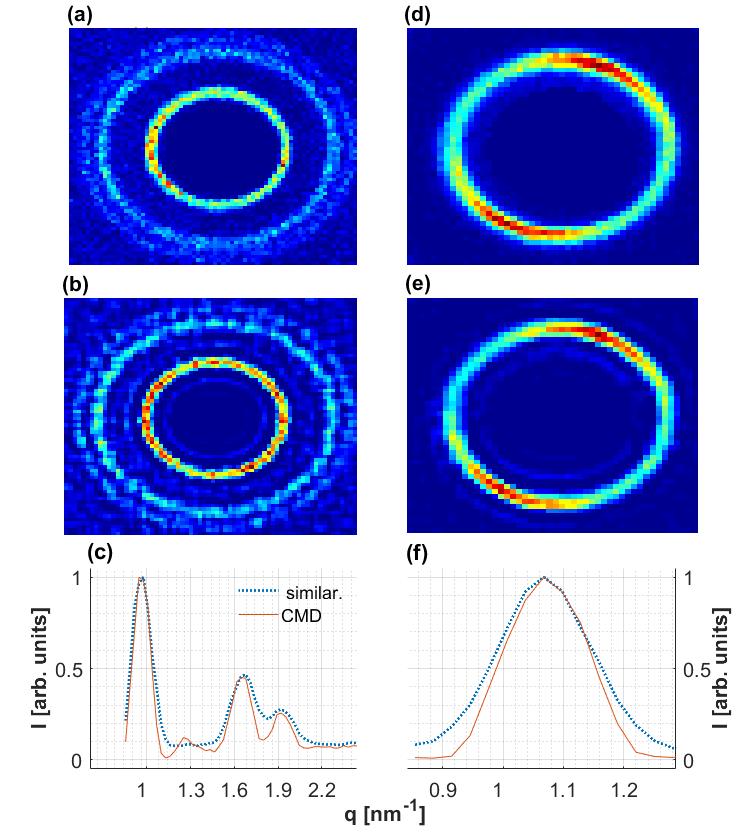}
\caption{CMD algorithm applied to measured data: (a-c) DOPE, (d-f) AgBh samples. 
(a, d) Measurements with PSF size of $0.6\times 0.6 mm^2$ with 3 minute exposure time. 
(b, e) CMD reconstruction using 6 different PSFs each with 10 seconds exposure time.
(c, f) 1D azimuthal integration SAXS signal, presenting enhanced resolution.}
\label{fig:true_CMD}
\end{figure}

Furthermore, we studied the CMD algorithm limitations by simulating two concentric scattering peaks (as described in Sec. \ref{Synthetic_data_generation}), with widths of 0.7 pixels. First, we fixed one of the scattering rings to a radius of 30 pixels and reconstructed using CMD when the second ring was varied between 31 to 50 pixels (i.e., $\Delta q = 1 - 20$). Using our separation  criteria (Eq. \ref{eq:resolution_crit}), we find that the CMD algorithm approach the theoretical limits (dashed line) for $\Delta q$ larger than 9 pixels, and with reasonable reconstruction resolution for $\Delta q$ larger than 4 pixels (Fig. \ref{fig:separation_boundry}).  Importantly, the deviation from the theoretical $\delta(\Delta q \rightarrow 0)$ is unphysical, since, for a practical system with finite pixel size, all reconstructions attempts will fail below the separation of several pixels. Reconstructions with various inner circle radii (10, 30, and 50 pixels) resulted in similar conclusions.

\begin{figure}
\includegraphics[width=1\textwidth]{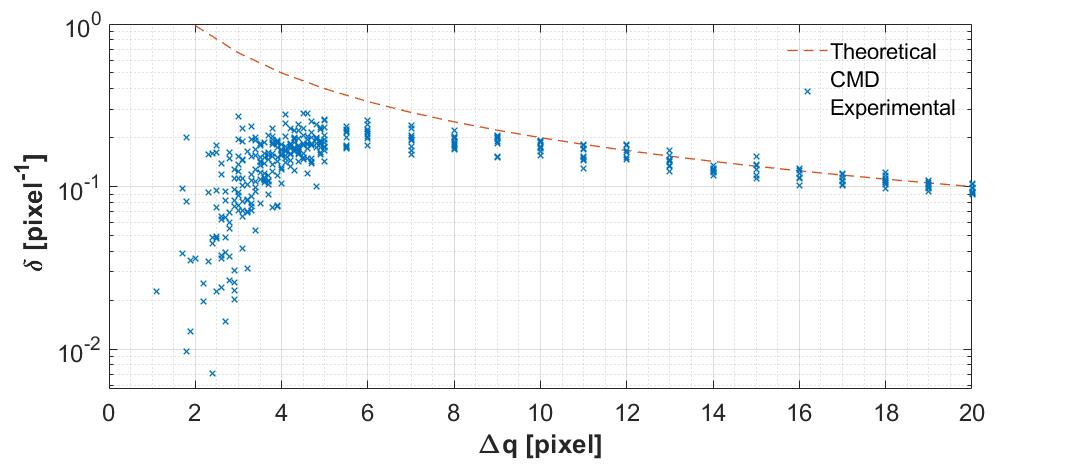}
\caption{CMD's resolution capabilities for separating nearby scattering rings. Separation  criteria (Eq. \ref{eq:resolution_crit}), $\delta$, for various distance between the rings ($\Delta q$). Red dashed-line, is the theoretical value of $\delta$ calculated directly from the GT image.}
\label{fig:separation_boundry}
\end{figure}

Last, we examine how the number of different PSFs influences reconstruction capabilities. For a given number of PSFs used by the CMD, we evaluated the best performing subset of PSFs (out of a pool of 16 different PSFs). As shown in Fig. \ref{fig:CMD_params}c, increasing the number of measurements results in better reconstructions. However, increasing the set size comes at the expense of increasing the total measurement time. Thus a harsh comparison is an equivalent "exposure time" without any PSF blurring (Figure \ref{fig:CMD_params}c, dashed lines). Our results clearly demonstrate the supremacy of the CMD reconstructions over a large parameter set, and in particular, for short total exposure times, where our reconstruction error is significantly smaller than the unrealistic point-like PSF. As exposure time increases, CMD reconstructions converge to the optimal point-like PSF.  

Motivated by our simulated data, we proceed on measuring the above-mentioned samples and evaluating the CMD performance (Fig. \ref{fig:true_CMD}). For PSF details see Appendix \ref{app:measured_PSF}. In order to optimize the reconstruction, we conducted a grid search over the expected $\lambda$ and $\sigma$ values of similar patterns studied with synthetic data. We find that $\lambda=0.04$ and normal distribution probability function of $\sigma$ centered about 70\% of the mean intensity of the smallest PSF used and with a width of 20\% larger.

The CMD results show significant improvement in separation criteria from $\delta=0.57$ in the measured image (Fig. \ref{fig:true_CMD}a) to $\delta=1.15$ in the reconstructed image at the output of CMD (Fig. \ref{fig:true_CMD}b). The comparison was made to equivalent scattering time with the smallest PSF used by the CMD. The reconstruction is particularly encouraging, enabling to distinguish between neighboring scattering peaks.

\subsection{Super-resolution SAXS Reconstruction}
We are now in a position to combine the two aforementioned resolution enhancement approaches to reconstructed super-resolution SAXS (srSAXS) 2D images. The algorithms and examples are deposited in public repository \cite{codeGitHub}. Given that we took $m$ images with different PSF, and we have done so with resolution enhancement factor $f$, the simplified solution is to perform SR and CMD procedures separately.

Both from a physical and numerical point of view, the SR procedure should take place first. Since the blurring arises from the not-point-wise nature of the source, the blurring occurs before the decimation, occurring on the camera. Hence, the reconstruction should be in the reverse order. From a numerical point of view, performing the CMD first corrupts the sub-pixel translations, and therefore the registration step in the SR procedure is unable the align the images to one another. This was indeed confirmed in numerical evaluations.

In Fig. \ref{fig:srsaxs_sim} we demonstrate the added value of the full srSAXS approach on simulated data having two nearby peaks, which are inseparable at the decimated blurred image. However, post-processing the full srSAXS algorithm, the nearby peaks can be clearly identified with an added resolution to the combined SR and CMD approach $\delta_{srSAXS} = 0.026$. Remarkably, the srSAXS reconstruction is superior to the images simulated with the smallest PSF and equivalent total exposure time used $\delta_{Similar}=0.022$ (Fig. \ref{fig:srsaxs_sim}, gray dashed line).    

\begin{figure}
\includegraphics[width=1\textwidth]{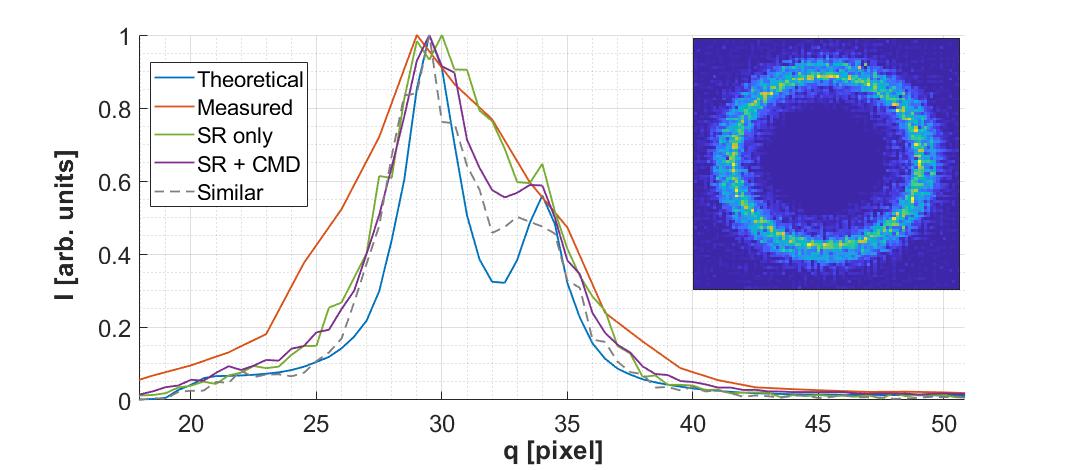}
\caption{SrSAXS reconstruction on simulated data. 
1D azimuthal integration of synthetic data simulating two close circles. The addition of CMD (purple line) improves the reconstruction resolution over the SR data alone (green line) and even the SR data applied on the finest PSF (gray dashed line). 2D reconstruction is presented in the inset.}
\label{fig:srsaxs_sim} 
\end{figure}

The combined srSAXS approach is further applied to our experimental data. In Fig. \ref{fig:true_data_srSAXS} we present our srSAXS reconstructions for the scattering data of DOPE and AgBh measured at 9 different subpixel positions ($f=3$) and using $m=6$ different PSFs. These results manifest the added value of the techniques combined. We do notice that the CMD produces additional satellite minors peaks. One way of solving this is by choosing other regularizing method, such as L1 or some Tikhonov Regularization \cite{Golub:2013}. Our experiments showed different results in this aspect but not with great significance and this remains to be further researched.
Evaluating the separation criteria for DOPE (between $q_{H(1,1)}$ and $q_{H(2,0)}$ peaks) we find $\delta_{SR}=0.48$, and $\delta_{srSAXS}=1.22$, for SR only and full srSAXS reconstructions, respectively. In comparison, for similar exposure time using the smallest PSF ($ 0.6  \times 0.6 mm^2$) used by the CMD, we find $\delta_{similar}=0.59$. For smallest slits opening of $ 0.2 \times 0.2 mm^2$, with twice the total exposure time  $\delta_{GT}=1.21$, comparable to our reconstructed capabilities. Similar enhanced resolution is also demonstrated for the AgBh data where the (001) peak width shrinks for $\sim 20\%$ using our methodology.

\begin{figure}
\includegraphics[width=1\textwidth]{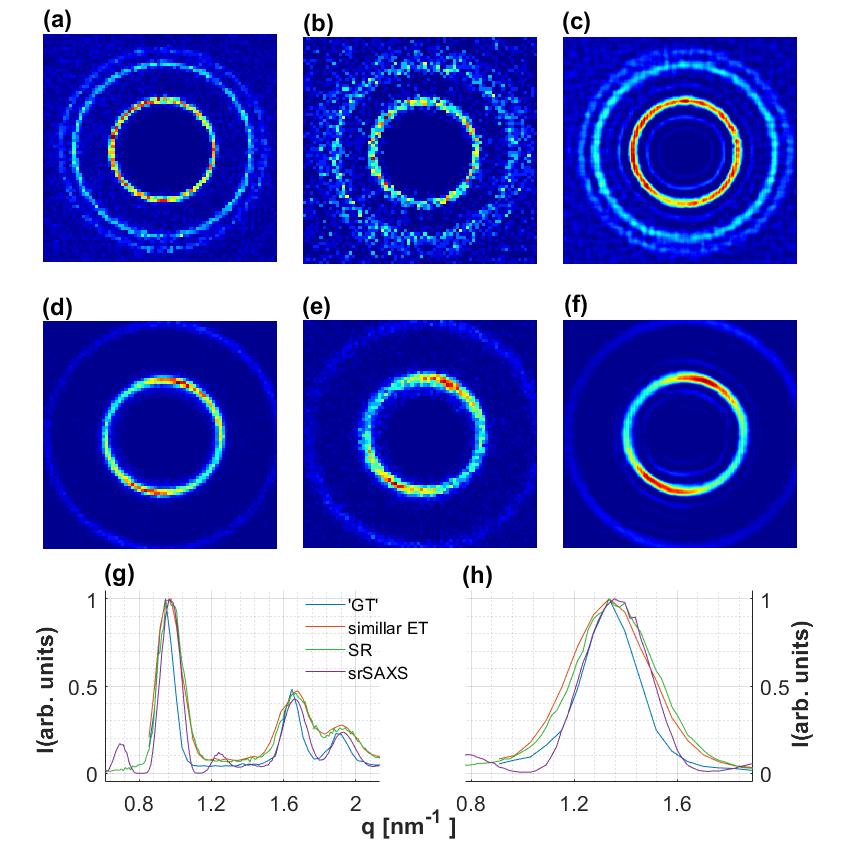}
\caption{srSAXS measurements and reconstruction for (a-c, g) DOPE and (d-f, h) AgBh. 
(a, d) "GT" image taken with exposure time of 15 minutes and the finest PSF ($0.2 \times 0.2 mm^2$).
(b, e) Measurements with exposure time of 10 sec and PSF of ($0.6 \times 0.6 mm^2$).
(c, f) srSAXS reconstructed images with f=3 and 6 different PSFs (detailed in Appendix \ref{app:measured_PSF}).
(g, h) 1D azimuthal integration signals. Blue lines are for "GT" as in (a, d), red lines are measurements of similar exposure time ($9\times 6 \times 10$ seconds), green lines are SR reconstructions, and purple lines are full srSAXS reconstructions.}
\label{fig:true_data_srSAXS}
\end{figure}

Finally, we compared our srSAXS algorithms with other reconstruction techniques applied to the measured data (DOPE, see Fig. \ref{fig:comparison}). As a common deconvolution technique, we implemented Richardson-Lucy (RL) algorithm  \cite{Lucy:1974,Richardson:1972} and gained $\delta_{RL}=1.19$  (Fig. \ref{fig:comparison}b and red line in Fig. \ref{fig:comparison}e). Additional comparison is presented using the full implementation of Farsiu et. al algorithm (FA) including the proposed deconvolution \cite{Farsiu:2004} (Fig. \ref{fig:comparison}d and green line in Fig. \ref{fig:comparison}e). The FA  results in $\delta_{FA}=1.17$ with low background signal. Much longer exposure time with minimal PSF has very good performance due to high SNR and negligible blur. However, among the different algorithms examined, our srSAXS showed better performance in separating nearby features with $\delta_{srSAXS}=1.22$ (Fig. \ref{fig:comparison}c and purple line in Fig. \ref{fig:comparison}e).

\begin{figure}
\includegraphics[width=0.8\textwidth]{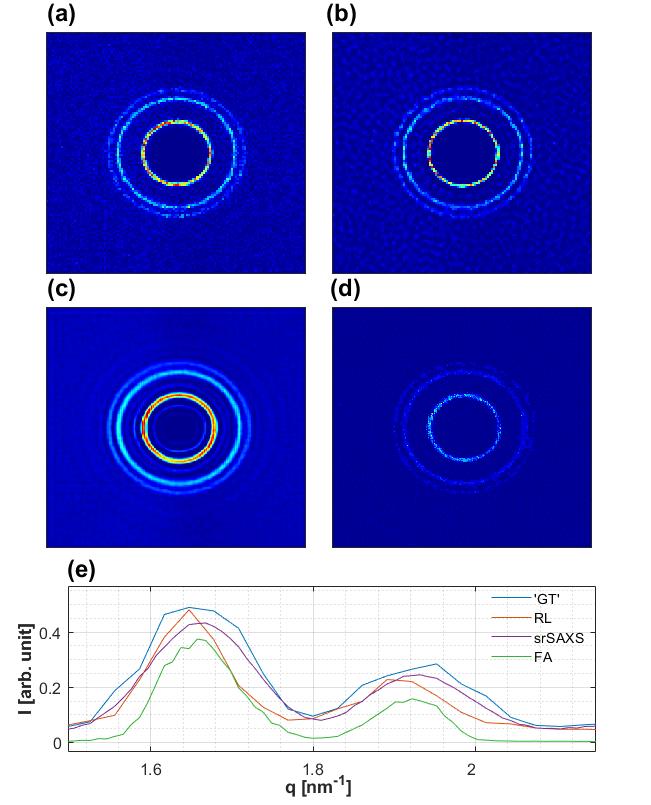}
\caption{Comparison of different reconstruction techniques.
(a) 'Ground truth' approximation from 15 minutes exposure time with $0.2 \times 0.2 mm^2$ PSF, (b) Richardson-Lucy deconvolution technique, using the $0.6 \times 0.6 mm^2$ PSF and 9 minute exposure time. (c) Our srSAXS technique with $6 \times 9 \times 10 seconds$ total exposure time (as in Fig. \ref{fig:true_data_srSAXS}c). (d) Full implementation of Farsiu algorithm (FA) \cite{Farsiu:2004} using the $0.6 \times 0.6 mm^2$ PSF and $9 \times 1$ minutes exposure time. (e) 1D azimuthal integration.}
\label{fig:comparison}
t\end{figure}

\section{Summary}

We demonstrate a new computationally and technically efficient method that significantly enhances SAXS angular resolution. For a limited photon flux, as in the case of lab-based systems, and a limited total experimental time, the recorded SAXS resolution is limited by low SNR. We demonstrate, both on synthetic and experimental measured data, two resolution enhancement procedures. In the first method, super-resolution is achieved by measuring the scattering signal for altered sub-pixel positions of the detector. For the second method, several exposures are taken using different PSFs. Each of the techniques resulted in enhanced resolution, while the best performing reconstruction is achieved when both techniques are applied one after the other.

\section{Acknowledgements}

We thank Guy Jacoby, and Micha Kornriech for helpful discussions and assistance with the measurements. We acknowledge the support of Multi-Dimensional Meteorology (MDM) consortium by the Israel Innovation Authority, and the Israeli Science Foundation (award number 550/15).

\bibstyle{iucr}
\referencelist{}

\appendix

\section{Synthetic PSF}
\label{app:syn_PSF}
As described in Section \ref{Synthetic_data_generation}, the blurring PSF were modeled with two error-functions at each axis:

\begin{equation} \label{eq:synt_PSF}
   f(x) = erf(\frac{a_1+x}{a_2}) + erf(\frac{a_1-x}{a_2}),
\end{equation}
where $a_1$ is a width parameter, $a_2$ is the 'sharpness' parameter and $x$ belong to a symmetrical interval around 0. This 'beam profile' was applied in $x$ and $y$ axes independently to produce 2D PSFs.
Throughout the paper, four such profiles were simulated, resulting in 16 different PSFs, which can be found in Fig \ref{fig:syn_PSF}. The parameters used to generate these PSFs are summarized in Table \ref{table:PSF}, with $x \in [-8,8]$ equaly spaced with 19 points.

\begin{table}
\centering
\label{table:PSF}
\begin{tabular}{l|l} 

 $a_1$ & $a_2$ \\ [0.5ex] 
 0.1 & 1  \\ 

 1.6 & 1.5  \\

 3.1 & 2 \\

 4.6 & 2.5 \\ 

\end{tabular} \\

\caption{Parameters used to generate synthetic PSF using Eq. \ref{eq:synt_PSF}}
\end{table}

\begin{figure}
\includegraphics[width=0.5\textwidth]{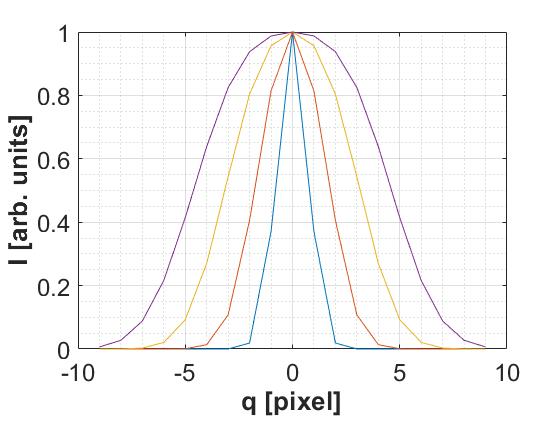}
\caption{Different profiles used for generating synthetic PSF}
\label{fig:syn_PSF}
\end{figure}

\section{Conjugate gradients details and mathematical proof}
\label{apn:conjugate_proff}
Using the notation presented in section \ref{sec:CMD}, we will show that the conjugate gradients method can be applied to our reconstruction problem with specific optimum conditions. Below, we find $\mathcal{P}^T$, later to be used in a conventional conjugate gradient method.

\begin{definition} 
\label{def:rot2}
Let $A\in\mathbb{R}^{n \times m}$, then $[\breve{A}]_{i, j} =: [A]_{n-i, m-j}$
\end{definition}

\begin{lemma}
Defining $\mathcal{P} = [P_1 * \quad P_2 * \quad...\quad P_m* \quad \sqrt{\lambda}I]^T$, $\mathcal{P}^T$ takes the form of $\mathcal{P}^T\mathcal{Y} =: \mathcal{Y}^T * \breve{\mathcal{P}}$
\end{lemma}
 
\begin{proof}
We prove the above for $m = 1$ and the extension to a general $m$ is trivial. recalling that $||x||_F^2 = \langle x, x \rangle_F = Tr(x^T x) = \sum_i [x^T x]_{ii}$,  we need to show that
$\sum_i [(\mathcal{P} x)^T  \mathcal{Y}]_{ii} = \langle \mathcal{P} x,  \mathcal{Y} \rangle_F = \langle x, \mathcal{P}^T \mathcal{Y} \rangle_F = \sum_i [x^T  (\mathcal{Y}^T * \breve{\mathcal{P}})]_{ii}$.

Therefore, 

$\langle \mathcal{P} X , \mathcal{Y} \rangle_F = 
\sum_i [(\mathcal{P} X)^T \mathcal{Y}]_{ii} = 
\sum_i (\sum_j [\mathcal{P} X]^T_{ij} Y_{ji})= 
\sum_i (\sum_j [P *  X]^T_{ij} Y_{ji}) = 
\sum_{ij} [P * X]_{ji} Y_{ji} = 
\sum_{ij} [ \sum_{kl} P_{k-j,l-i}  X_{kl}] Y_{ji} = 
\sum_{ij}  \sum_{kl} P_{k-j,l-i} Y_{ji} X_{kl} = 
\sum_{kl} \sum_{ij} \breve{P}_{j-k,i-l} Y_{ji} X_{kl} = 
\sum_{kl} X_{kl} [\sum_{ij} \breve{P}_{j-k,i-l} Y_{ji}] = 
\sum_{kl} X_{kl}  [ \breve{P} * Y]_{kl} =
\sum_l (\sum_k [{X_{lk}}]^T  [Y * \breve{P}]_{kl}) = 
\sum_l [X^T ( \mathcal{\breve{P}} \mathcal{Y})]_{ll} = 
\langle X ,  \mathcal{P}^T \mathcal{Y} \rangle_F $

\end{proof}
Applying the above to a standard conjugate gradients algorithm finds the optimal $X$ in Eq. \ref{eq:normal_equation}.

\section{Measured PSF}
\label{app:measured_PSF}
Conducting measurements included measuring the AgBh and DOPE samples as well as the PSF used for these measurements.
Like the samples, the PSF were measured 9 times with translations in 2 axes (f=3). The PSFs were determined by the following slit size (measured in $mm^2$) : \\
$0.2\times0.2$ ;
$0.6\times0.6$ ; $06\times0.8$ ; $0.6\times1$
$0.8\times0.6$ ; $08\times0.8$ ; $0.8\times1$

For each sample, we measured an approximation to the ground-truth using the finest PSF ($0.2\times0.2$) with longer exposure times (15 minutes). The rest of the scattering images were measured for 10 seconds intervals.

\end{document}